\documentclass[10pt, twocolumn, conference]{IEEEtran}

\usepackage{amsmath, mathrsfs,amsfonts}
\usepackage{amssymb}
\usepackage{graphicx}
\usepackage{cite}
\usepackage{algorithmic}
\usepackage{algorithm2e}
\usepackage{array}
\usepackage[tight]{subfigure}
\usepackage[top=0.74in, bottom=0.74in, left=0.55in, right=0.55in]{geometry}

\usepackage{enumerate}

\newtheorem{proof}{Proof}

\newtheorem{theorem}{Theorem}

\newtheorem{corollary}{Corollary}

\begin{document}
\title{Distributed Frequency Control in Power Grids \\ Under Limited Communication \thanks{This work was supported by DTRA grants HDTRA1-13-1-0021 and HDTRA1-14-1-0058.}}

\author{\IEEEauthorblockN{Marzieh Parandehgheibi \IEEEauthorrefmark{1}, Konstantin Turitsyn \IEEEauthorrefmark{2} and Eytan Modiano \IEEEauthorrefmark{1}}
\IEEEauthorblockA{\IEEEauthorrefmark{1}Laboratory for Information and Decision Systems, Massachusetts Institute of Technology, Cambridge, MA, USA\\
\IEEEauthorrefmark{2}Department of Mechanical Engineering, Massachusetts Institute of Technology, Cambridge, MA, USA\\}}

\maketitle 

\begin{abstract}
In this paper, we analyze the impact of communication failures on the performance of optimal distributed frequency control. We consider a consensus-based control scheme, and show that it does not converge to the optimal solution when the communication network is disconnected. We propose a new control scheme that uses the dynamics of power grid to replicate the information not received from the communication network, and prove that it achieves the optimal solution under any single communication link failure. In addition, we show that this control improves cost under multiple communication link failures. 

Next, we analyze the impact of discrete-time communication on the performance of distributed frequency control. In particular, we will show that the convergence time increases as the time interval between two messages increases. We propose a new algorithm that uses the dynamics of the power grid, and show through simulation that it improves the convergence time of the control scheme significantly. 
\end{abstract}

\section{Introduction}
The main objective of a power grid is to generate power, and transmit it to the consumers. The power grid balances supply and demand  through frequency control. This is done both at the local (generator) level, and the wide-area level as follows.
\begin{enumerate}
\item \textbf{Primary Frequency Control (Droop Control):} A local frequency controller that balances the power by speeding up or slowing down the generators; i.e. creating deviation from the 60Hz nominal frequency; this controller responds to the changes in power within milliseconds to seconds.
\item \textbf{Secondary Frequency Control (AGC):} AGC is a centralized frequency controller that re-adjusts the set points of generators to balance the power and restores the nominal frequency; this is a close-loop automatic controller that is applied every 2-4 seconds and requires communication network between AGC and generators.
\item \textbf{Economic Dispatch:} This is a centralized controller that reschedules the generators to minimize the cost of generation; this control decision is made by the ISO every 10-15 minutes, and requires communication network between ISO and generators.
\end{enumerate}

The future power grid is going to integrate renewable energy resources. This will increase the fluctuations in the generation, and requires more reserve capacity to balance the power. One of the approaches to balancing power without having large reserve capacities is demand response, where loads are ``adjustable'' and participate in balancing the power. Since the number of loads is large, they cannot be controlled in a centralized manner. Thus, it is essential to use ``distributed'' control for demand response that incorporates all three stages of traditional frequency control.

Recently, there have been many attempts to develop distributed frequency control mechanisms. In \cite{dominguez2010coordination}, the authors consider the case that the total amount of required power is known, and designed a distributed algorithm that determines the amount of load participation to minimize the cost. In \cite{andreasson2013distributed}, the authors design a distributed frequency controller which balances the power under unknown changes in the amount of generation and load, and compare its performance with a centralized controller. 

In \cite{simpson2013synchronization}, the authors propose a primary control mechanism, similar to the droop control, for microgrids leading to a desirable distribution of power among the participants, and propose a distributed integral controller to balance the power. These results are extended in \cite{dorfler2014breaking}, where the authors use a similar averaging-based distributed algorithm to incorporate all three stages of frequency control in microgrids. Moreover, in \cite{zhao2015distributed}, the authors propose a similar consensus-based algorithm for optimal frequency control in transmission power grid.

In \cite{zhao2014design}, the authors use a primal-dual algorithm to design a primary frequency control for demand response in power grid. The results are extended in \cite{mallada2014distributed} and \cite{mallada2014optimal}, where the authors design a primal-dual algorithm to model all three stages of a traditional frequency control in the power grid.

Although there exist several different distributed frequency control mechanisms in the literature, they all rely on the use of communication to exchange control information (e.g., Lagrangian multipliers). Moreover, convergence to an optimal solution requires the underlying communication network to be connected. In addition, in the design and analysis of all these controllers, it is assumed that the communication messages between neighboring nodes are transmitted in continuous time; however, in practice, these messages will be transmitted in discrete time. In this paper, we analyze the performance of a consensus-based control scheme under communication failures. We show that when the communication network is disconnected, the control scheme balances the power by retrieving the normal frequency; however, its cost is not optimal. Moreover, we analyze the effect of discrete-time communication on the convergence time of this control scheme. 

Next, we propose a novel control algorithm which uses the information from the power flow to replicate the direct information received from the communication network. We prove that our algorithm achieves the optimal solution under any single communication link failure. We also show via simulation results that our algorithm improves the cost under multiple communication failures. Finally, we propose a sequential algorithm based on our control mechanism, and show that it improves the convergence time under discrete-time communication.

The rest of this paper is organized as follows. In Section \ref{Model_Sec}, we describe the power grid's model. In Section \ref{Distributed_Sec}, we describe a consensus-based distributed frequency control, and analyze it under communication link failures and discrete-time communication. In Section \ref{TwoBus_Sec}, we will propose a novel decentralized control for a two-node system and prove its optimality and stability, and in Section \ref{CommFailures_Sec}, we extend our control mechanism for multi-node systems under disconnected communication networks. Next, in Section \ref{DiscreteComm_Sec}, we propose a sequential control algorithm that improves the convergence time under discrete-time communication. Finally, we conclude in Section \ref{Conclusion_Sec}. 

\section{System Model}\label{Model_Sec}

Let $\mathcal{G_P=\{N_P,E_P\}}$ be the power grid, where $\mathcal{N_P}$ denotes the set of power nodes, and $\mathcal{E_P}$ denotes the set of power lines. The power at every node $j$, whether it is a generator or a load, consists of adjustable and unadjustable parts. The unadjustable part is the amount of power that cannot be changed; i.e. fixed demand or generation. The adjustable part is the amount of power that can be changed; i.e. controllable load or generation. The sum of the total power determines the amount of power imbalance in the grid, which leads to the frequency deviation. The role of a controller is to balance the power by using the adjustable power at all nodes with minimum cost. Next, we describe the dynamics of the power grid which translate the power imbalance to frequency deviation. Then, we describe the optimal control policy. 

Let $M_j$ be the inertia of node $j$, and $D_j$ be the droop coefficient of node $j$. Moreover, let $p_j(t)$ be the unadjustable power and $u_j(t)$ be the adjustable power (control) at node $j$ and at time $t$. In addition, let $B_{jk}$ be the susceptance of power line $(j,k)$, and $f_{jk}(t)$ be the amount of power flow from node $j$ to node $k$ at time $t$. We can describe the dynamics of the power grid using the swing equation at every node and the power flow equation at every line as follows. 

\vspace{-3mm}
\begin{subequations}
\begin{alignat}{3}
	 & M_j \dot{\omega_j}(t) = -D_j \omega_j(t) + p_j(t) + u_j(t) &&-\sum_{k: (j,k)\in \mathcal{E_P}}f_{jk}(t)  \nonumber \\
	 & \quad\quad && j \in \mathcal{N_P} \\
	 & \dot{f_{jk}}(t) = B_{jk}(\omega_j(t)-\omega_k(t)) \quad&&  (j,k)\in \mathcal{E_P}
\end{alignat}
\label{grid_dynamics}
\end{subequations}
\vspace{-3mm}

The objective of our control is to minimize the total cost of adjustable power at steady-state while balancing power. Let $p_j^*$ be the steady-state unadjustable power, and $u_j^*$ be the steady-state adjustable power (control) at node $j$. Moreover, let $f^*_{jk}$ be the steady-state power flow from node $j$ to node $k$. The optimal steady-state control can be formulated as follows.

\vspace{-3mm}
\begin{subequations}
\begin{alignat}{4}
	& \min_{u^*,f^*} \quad && \sum_{j\in \mathcal{N_P}} \frac{1}{2}C_j u_j^{*2}  \\
	& s.t. && p_j^*+u_j^*-\sum_{k: (j,k)\in \mathcal{E_P}} f^*_{jk}=0 \quad j \in \mathcal{N_P}
\end{alignat}
\label{cost_opt_eq}
\end{subequations}
\vspace{-3mm}

It was shown in \cite{zhao2015distributed} and \cite{dorfler2014breaking} that the optimal solution to equation (\ref{cost_opt_eq}) has the form of $C_i u_i^*=C_j u_j^*$, where $\sum_{j \in \mathcal{N}} u_j^*=-\sum_{j \in \mathcal{N}} p^*_j$ \footnote{Note that satisfying condition $C_i u_i^*=C_j u_j^*$ can also be interpreted as fairness in sharing the loads.}.

\section{Distributed Control}\label{Distributed_Sec}

Let the power grid be supported by a connected communication network $\mathcal{G_C=\{N_C,E_C\}}$, where $\mathcal{N_C}$ denotes the set of communication nodes, and $\mathcal{E_C}$ denotes the set of communication links. The optimal distributed frequency control can be described by the following differential equation.

\vspace{-3mm}
\begin{equation}
	C_i \dot{u_i}(t) = -\omega_i(t) - C_i \sum_{j: (i,j) \in \mathcal{E_C}}(C_i u_i(t) - C_j u_j(t)) \quad i \in \mathcal{N_P}  \label{Dist_cont_eq}
\end{equation}
\vspace{-3mm}

Accordingly, the distributed control works as follows: node $i$ measures the local frequency $\omega_i$, receives the information $C_j u_j(t)$ from the neighbor nodes via the communication network, and updates the local control value $u_i(t)$. It is shown in \cite{zhao2015distributed} and \cite{dorfler2014breaking} that if the communication network is connected, the control mechanism in equation (\ref{Dist_cont_eq}) converges to the optimal solution, which is globally asymptotically stable.

\subsection{Impact of Communication Link Failures}
The control mechanism in equation (\ref{Dist_cont_eq}) will achieve the optimal solution if the communication network is connected. However, if the communication network is disconnected, while power will be balanced, optimal cost may not be achieved; i.e. it cannot guarantee that $C_iu_i^*=C_ju_j^*$ for all $i,j$ nodes. Next, we show via an example that the impact on the cost could be significant. 

Consider the power grid in Figure \ref{Toy_Example} (The data of the grid and the costs can be found in Appendix \ref{ToyPower_Data}). In this example, the communication network has the same topology as the power grid. The total load in this grid is 25 p.u., and we increase the load in node 3 by 5 p.u. ($20\%$ total increase). Simulation results show that the optimal cost, by applying control mechanism \ref{Dist_cont_eq} under a fully connected communication network, is 23.27. If the communication link between nodes $2$ and $7$ fails, the cost increases to $35.69$, while the cost under no communication is $39.11$. This example shows that the failue of only one communication link could have a significant impact on the cost of distributed control.

\begin{figure}[ht]
\centering
\includegraphics[scale=0.4]{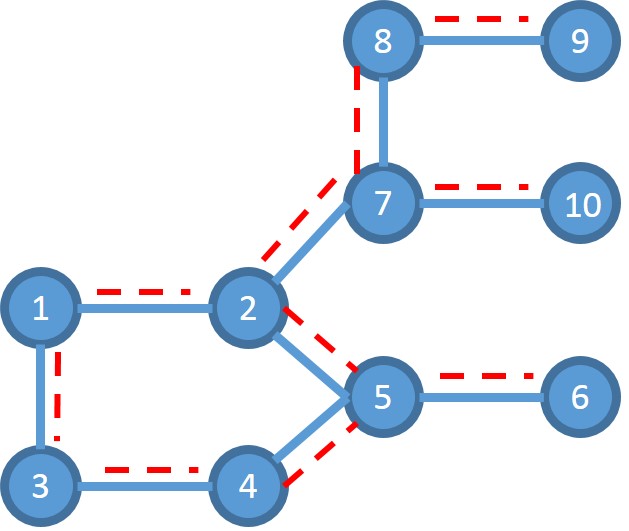}
\caption{Power Grid Toy Example - Solid lines are power lines and dashed lines are communication lines.}
\label{Toy_Example}
\end{figure}

\subsection{Impact of Discrete-time Communication}\label{Impact_Disc_Comm}
In the design and analysis of the distributed control mechanism described in equation (\ref{Dist_cont_eq}), it is assumed that the communication messages are updated in continuous time. However, in reality, the communication messages will be updated in discrete time. Let $T$ be the time interval between two communication messages. Then, the distributed control can be described as follows.

\begin{align}
	 & C_i \dot{u_i}(t) = -\omega_i(t) - C_i \sum_{j: (i,j) \in \mathcal{E_C}}(C_i u_i(t) - C_j u_j(KT)) \nonumber \\
	 & \quad\quad\quad\quad\quad\quad\quad\quad i \in \mathcal{N_P},  \quad KT \leq t \leq (K+1)T  \label{Discrete_Dist_cont_eq}
\end{align}

Define the convergence time $t^*$ to be the first time such that $|(Cost(t^*)-Cost^*)|<0.01$, where $Cost(t^*)$ is the cost at time $t^*$ and $Cost^*$ is the optimal cost. By running the control in equation (\ref{Discrete_Dist_cont_eq}) on the power grid in Figure \ref{Toy_Example} for different values of $T$, it can be seen that the time of convergence increases as $T$ increases (See Figure \ref{ConvergenceTime}).

\begin{figure}[ht]
\centering
\includegraphics[width=0.6\linewidth]{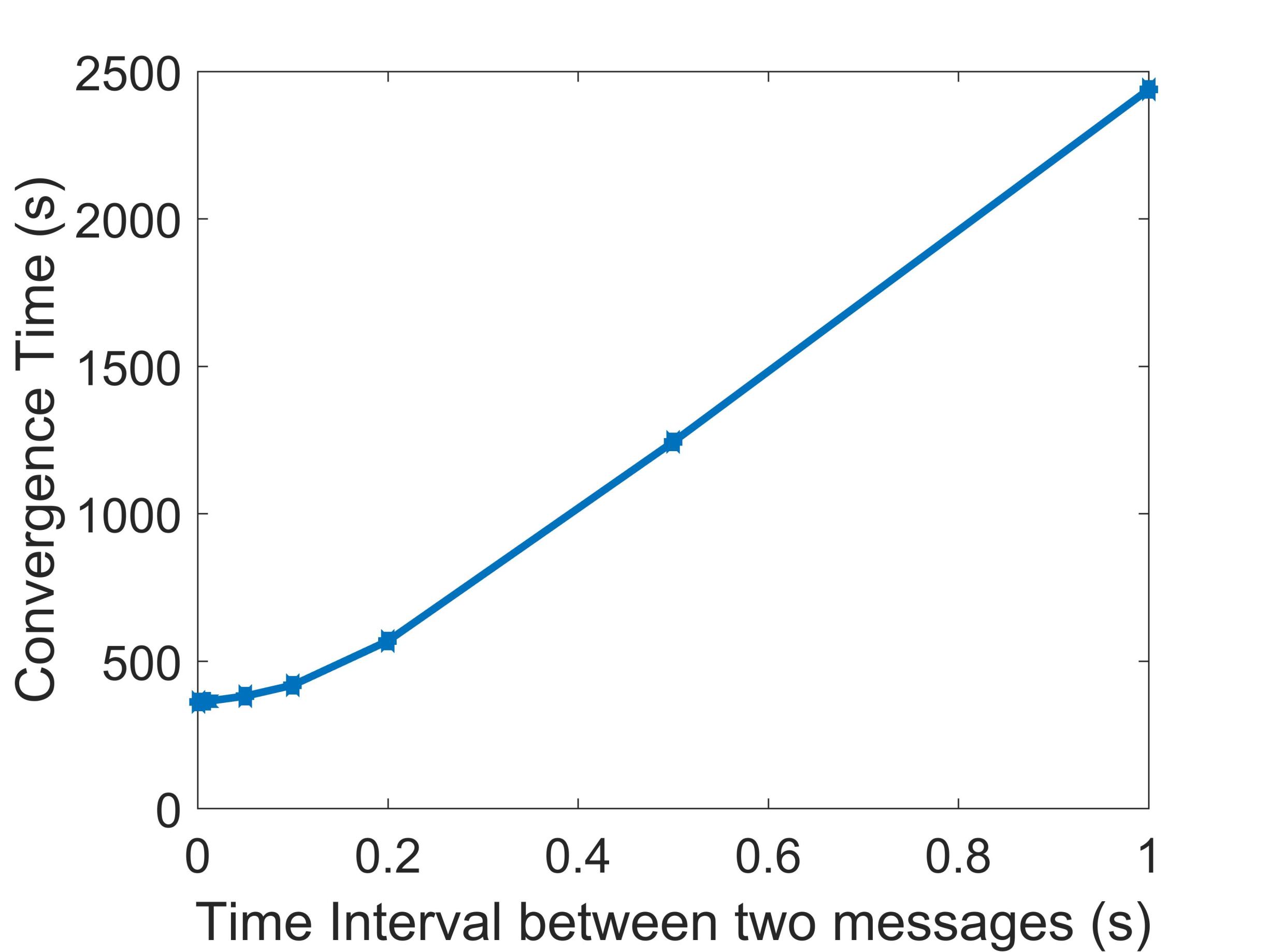}
\caption{Convergence Time increases as $T$ increases.}
\label{ConvergenceTime}
\end{figure}

\section{Decentralized Control for Two-node System}\label{TwoBus_Sec}

In this section, we consider a two-node system connected by a power line and communication link as in Figure \ref{TwoBus_Comm}. As described, when the communication link fails, node $i$ does not receive information $C_ju_j(t)$ and node $j$ does not receive information $C_i u_i(t)$. Therefore, the optimal cost cannot be achieved. Next, we propose a control algorithm that uses the dynamics of the power grid instead of direct information $C_i u_i(t)$ and $C_j u_j(t)$, and still achieves the optimal solution. 

Previously, the adjustable power at both nodes $i$ and $j$ was updated based on the local frequency and the information received from the neighboring node. In our control scheme, we update the adjustable power at every node based on the local frequency and a local artificial variable, where this variable is updated based on the power flow dynamics between the two nodes. Since the changes in the flow is a function of the frequency at both nodes, it contains some indirect information about the adjustable control as well as the cost of the neighbor node. We prove that this information is enough to guarantee the optimality of the our control scheme.

Let $q_{i}$ and $q_{j}$ be the two artificial variables at nodes $i$ and $j$, respectively. Our decentralized control for the two-node system can be described as follows.

\begin{subequations}
	\begin{alignat}{1}
 		& C_i \dot{u_i}(t) = -\omega_i(t) - q_{i}(t)   \label{update_Control1}\\
 		& C_j \dot{u_j}(t) = -\omega_j - q_{j}(t)   \label{update_Control2}\\
 		& \dot{q_{i}}(t)= - \frac{\dot{f_{ij}}}{B_{ij}} - 2q_{i}(t)   \label{update_Control3}\\ 		
 		& \dot{q_{j}}(t)=   \frac{\dot{f_{ij}}}{B_{ij}} - 2q_{j}(t)    	\label{update_Control4}	
 	\end{alignat}
 	\label{TwoBusControl}
\end{subequations}

As described, control at node $i$ is updated only based on the local frequency $\omega_i$ and the value of artificial variable $q_i$. Moreover, value of $q_i$ is updated based on the derivative of flow $f_{ij}$ which can be observed locally. Similarly, control at node $j$ depends on the local frequency $\omega_j$ and the derivative of flow $f_{ji}$ which can be observed locally. Thus, there is no need to a communication network between nodes $i$ and $j$. Next, we claim that the new control achieves the optimal solution (See Figure \ref{Two_Bus_Fig}).

\begin{figure}[ht]
\centering
\subfigure[Two-Node System with Communication]
{\label{TwoBus_Comm}\includegraphics[scale=0.45]{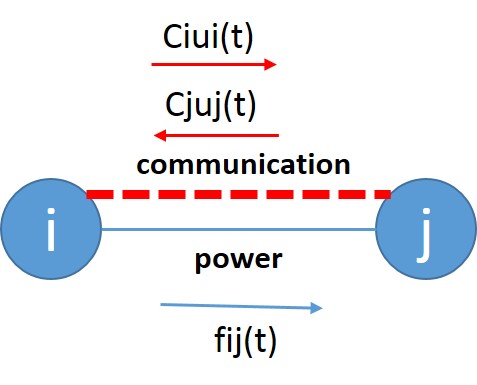}}\hfill
\subfigure[Two-Node System without Communication]
{\label{TwoBus_NoComm}\includegraphics[scale=0.45]{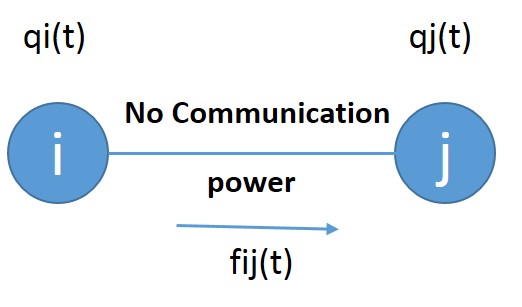}}
\caption{Let $t_0$ be the time failure: node $i$ knows $c_ju_j(t_0)$ and node $j$ knows $c_iu_i(t_0)$; Nodes $i$ and $j$ can initialize $q_{i}(t_0)$ and $q_{j}(t_0)$ properly to guarantee optimality}
\label{Two_Bus_Fig}
\end{figure}

Using the new control as in equations (\ref{TwoBusControl}), the dynamics of the system can be written as follows.

\begin{subequations}
	\begin{alignat}{1}
  	& M_i \dot{\omega_i}(t) = -D_i \omega_i(t) + p_i(t) + u_i(t) -f_{ij}(t)  \label{eq1} \\
		& M_j \dot{\omega_j}(t) = -D_j \omega_j(t) + p_j(t) + u_j(t) +f_{ij}(t)  \label{eq2}\\
	 	& \dot{f_{ij}}(t) = B_{ij}(\omega_i(t)-\omega_j(t)) \label{eq3}\\
 		& C_i \dot{u_i}(t) = -\omega_i(t) - q_{i}(t)   \label{eq4}\\
 		& C_j \dot{u_j}(t) = -\omega_j(t) - q_{j}(t)  \label{eq5}\\
 		& \dot{q_{i}}(t)= -(\omega_i(t) - \omega_j(t)) - 2q_{i}(t)   \label{eq6}\\ 		
 		& \dot{q_{j}}(t)= -(\omega_j(t) - \omega_i(t)) - 2q_{j}(t)    \label{eq7}
 	\end{alignat}
 	\label{TwoBusControl_rewrite}
\end{subequations}

In the following, we will prove the optimality and stability of the dynamical system described in equation (\ref{TwoBusControl_rewrite}).

\subsection{Optimality}

\begin{theorem}\label{Optimality_TwoBus}
Let $q_{i}(t_0) = -q_{j}(t_0) = C_i u_i(t_0) - C_j u_j(t_0)$. Then, the equilibrium point of the system described in equation (\ref{TwoBusControl_rewrite}) achieves the optimal cost.
\end{theorem}

\begin{proof}

In order to prove the optimality, we need to show that equation (\ref{TwoBusControl_rewrite}) guarantees $\omega^*_i=\omega^*_j=0$ and $C_i u^*_i=C_j u^*_j$ at the equilibrium point; i.e. power is balanced, and cost is minimized.

At the equilibrium, all of the derivatives in equation (\ref{TwoBusControl_rewrite}) are equal to zero. Therefore, we will have the following equations.

\vspace{-2mm}
\begin{subequations}
	\begin{alignat}{1}
		& \omega^*_i - \omega^*_j = 0\\
		& \omega^*_i - q^*_{i} =0\\
		& \omega^*_j - q^*_{j} =0\\
		& -(\omega^*_i - \omega^*_j) - 2q^*_{i} =0 \\
		& -(\omega^*_j - \omega^*_i) - 2q^*_{j} =0
 	\end{alignat}
	\label{equilibrium}
\end{subequations}
\vspace{-2mm}

Solving equations (\ref{equilibrium}) results in $\omega^*_i=\omega^*_j=0$, which guarantees that power is balances at the equilibrium point. In addition, we will have $q^*_{i}=q^*_{j}=0$.

Equations (\ref{eq6}) and (\ref{eq7}) show that $\dot{q_{i}}(t)=-\dot{q_{j}}(t)$ for all time $t\geq t_0$. Since we have initialized $q_{i}(t_0)= -q_{j}(t_0)$, it is easy to see that $q_{i}(t)=-q_{j}(t)$ for $t\geq t_0$. 

Next, we subtract equation (\ref{eq5}) from equation (\ref{eq4}). Thus, we will have $C_i \dot{u_i} - C_j \dot{u_j} = -(\omega_i-\omega_j) - 2q_{i}$ which is equal to the equation (\ref{eq6}). Therefore, $\dot{q_{i}} = C_i \dot{u_i} - C_j \dot{u_j}$. 

By taking integral over both sides from $t=t_0$ to infinity, we will have $(C_i u^*_i - C_i u_i(t_0))-(C_j u^*_j - C_j u_j(t_0)) = q^*_{i}-q_{i}(t_0)$, and assumption $q_{i}(t_0) = C_i u_i(t_0) - C_j u_j(t_0)$ results in $c_i u^*_i - C_j u^*_j = q^*_{i}$. Since $q^*_{i}=0$, $C_i u^*_i = C_j u^*_j$ which guarantees the optimality of the equilibrium point.

\end{proof}

\subsection{Stability}

Next, we prove that the equilibrium point of the dynamical system described in equation (\ref{TwoBusControl_rewrite}) is globally asymptotically stable. Since our dynamical system is linear, it is enough to show that the roots of the characteristic polynomial of the system are all located in the negative side of the plane.

Let $D, M, C \in R^{2 \times 2}$ be diagonal matrices denoting the droop coefficient, inertia and cost at nodes $i$ and $j$, respectively. Let $B$ be the susceptance of the power line between nodes $i$ and $j$. Moreover, let $A_p \in R^{2\times 1}$ be the node-edge incidence matrix, $L_p^B = A_p B A_p^T$ be the weighted laplacian matrix of the power grid, and $L_c \in R^{2 \times 2}$ be the laplacian matrix of the communication network. Finally, let $s(\lambda)$ be the characteristic polynomial of the system.

By applying the schur complement formula as well as elementary row operations, $s(\lambda)$ can be simplified to the following (See Appendix \ref{Two-Node-Simplify_Subsec} for more details).

$s(\lambda)= (\lambda+2) det(M^{-1}) det(H(\lambda))$ where $H(\lambda) =$
\begin{align}
	& \begin{bmatrix}	(\lambda^2 D + \lambda^3 M + \lambda C^{-1}) + (\lambda L_cD+\lambda^2 L_c M + (2+\lambda)L_p^B)	\end{bmatrix} \nonumber
\end{align}

Since the system is linear, it is enough to show that the real parts of all roots of characteristic polynomial $s(\lambda)$ are negative.

\begin{theorem}
The conditions in equation (\ref{sufficient}) are sufficient to guarantee that the equilibrium point of the system described in equation \ref{TwoBusControl_rewrite} is globally asymptotically stable.

\vspace{-5mm}
\begin{subequations}
	\begin{alignat}{1}
	&	M \succ 0  \label{cond1}\\
	&	\frac{1}{2}(L_c M+M L_c) + D \succ 0 \label{cond2} \\
	&	\frac{1}{2}(L_cD + DL_c)+L_P^B + C^{-1} \succ 0	 \label{cond3}\\
	&\lambda_{min}[(L_p^B + \frac{1}{2}(L_cD + DL_c) + C^{-1})]\times \nonumber \\
	&\lambda_{min}[(	\frac{1}{2}(L_c M+M L_c) + D)] > 4B\max\{M_1,M_2\} \label{cond4}
 	\end{alignat}
	\label{sufficient}
\end{subequations}
\vspace{-2mm}

\end{theorem}

\begin{proof}
See Appendix \ref{Two-Node-Stability_Subsec}.
\end{proof}

Next, we argue that sufficient conditions in equations (\ref{cond1})-(\ref{cond4}) often hold in practice. Condition (\ref{cond1}) holds as inertia is a positive value. Condition (\ref{cond2}) holds as the inertia of nodes in a distribution network is very small; and the matrix becomes strictly diagonally dominant. Conditions (\ref{cond3}) and (\ref{cond4}) hold if the cost values are scaled down; i.e. increase $C^{-1}$. Note that the only requirement for optimality of the control is that the ratio of power distribution be proportional to the inverse ratio of costs. Thus, scaling all the cost values will not affect the solution.

\section{Control under Communication Link Failures}\label{CommFailures_Sec}

In this section, we extend the idea in Section \ref{TwoBus_Sec} to multi-node systems. In particular, we introduce a new control mechanism that uses the dynamics of the power flow between adjacent nodes to replicate the direct information transmitted between them via a communication link. We show that our new control mechanism achieves the optimal solution under single communication link failure, and improves the cost under multiple communication link failures.

\subsection{Single Communication Link Failure}

Consider the power grid and communication network in Figure \ref{FailedComm_Single}. Suppose the communication link between nodes $i$ and $j$ fails. We claim that if nodes $i$ and $j$ update their local control decision only based on the power flow between nodes $i$ and $j$, and the rest of the nodes keep their previous control rule, the dynamical system will converge to the optimal solution. The new control mechanism can be described as follows.

\begin{subequations}
	\begin{alignat}{1}
		& C_k \dot{u_k}(t) = -\omega_k(t) - C_k \sum_{l: (k,l) \in \mathcal{E_C}}(C_k u_k(t) - C_l u_l(t)) \nonumber \\
		&\quad\quad\quad\quad\quad\quad\quad\quad\quad\quad\quad\quad\quad\quad\quad\quad k \in \mathcal{N}\backslash \{i,j\} \label{Original_control_singlefail}\\
 		& C_i \dot{u_i}(t) = -\omega_i(t) - q_{i}(t)     \label{updated_control_failure1}\\
 		& C_j \dot{u_j}(t) = -\omega_j(t) - q_{j}(t)   \\
 		& \dot{q_{i}}(t)= -(\omega_i(t) - \omega_j(t)) - 2q_{i}(t)   \\ 		
 		& \dot{q_{j}}(t)= -(\omega_j(t) - \omega_i(t)) - 2q_{j}(t)    \label{updated_control_failure4}
 	\end{alignat}
 	\label{FailedComm_Control}
\end{subequations}

\begin{figure}[ht]
\centering
\includegraphics[scale=0.4]{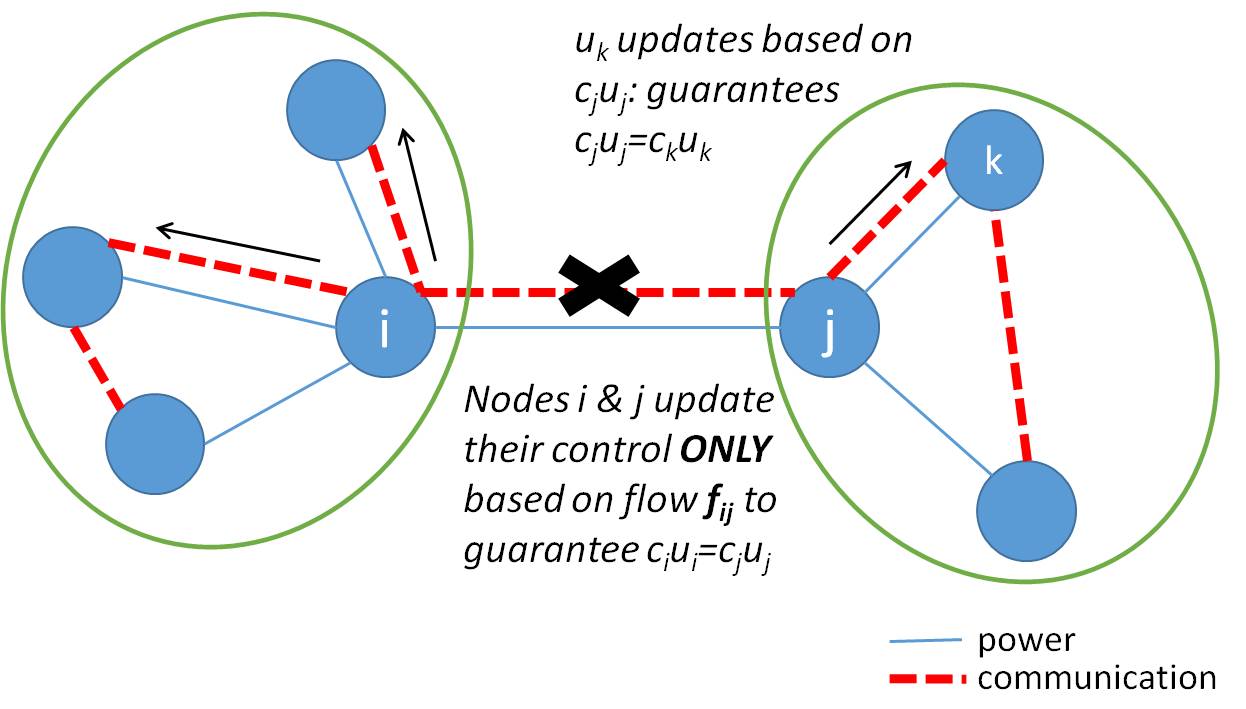}
\caption{Power Grid and Communication Network - Solid lines are power lines and dashed lines are communication lines.}
\label{FailedComm_Single}
\end{figure}

According to equation (\ref{FailedComm_Control}), all the nodes that are connected to node $i$ via the communication network, receive the information $c_i u_i(t)$ from node $i$; however, node $i$ does not update its control based on the information received from other nodes via communication network. Similarly, all the nodes connected to node $j$ via the communication network, update their control based on the information $C_ju_j$ they receive from node $j$; however, node $j$ does not use the information it receives from other nodes via the communication network. Instead, nodes $i$ and $j$ update their control only based on their local frequency and the power flow between nodes $i$ and $j$. This control rule can be interpreted as a master/slave algorithm, where nodes $i$ and $j$ are the master nodes that guarantee $C_i u_i^*=C_j u_j^*$, and the rest of nodes are the slave nodes that follow the changes in nodes $i$ and $j$. 

\begin{theorem}
Suppose the communication link between nodes $i$ and $j$ fails at time $t_0$, but they are connected via a power line. By updating the control mechanism according to equation (\ref{FailedComm_Control}), and initializing $q_{i}(t_0)= -q_{j}(t_0)=C_i u_i(t_0) - C_j u_j(t_0)$, the optimal solution will be achieved.
\end{theorem}
\begin{proof}
Equation (\ref{Original_control_singlefail}) guarantees that $C_k u_k^* = C_l u_l^*$ for all $ k \in \mathcal{N}\backslash \{i,j\}$. In particular, for any node $k$ connected to node $i$, $C_ku_k^*=C_iu_i^*$, and for any node $k$ connected to node $j$, $C_ku_k^*=C_ju_j^*$. On the hand, equations (\ref{updated_control_failure1})-(\ref{updated_control_failure4}) guarantee that $C_i u_i^* = C_j u_j^*$ (See Theorem \ref{Optimality_TwoBus} for optimality of a two-node system). Therefore, the equilibrium point is optimal.
\end{proof}

\begin{corollary}
Suppose the power grid has a connected topology, and the original communication network contains a subtree of the power grid. Then, the control mechanism described in equation (\ref{FailedComm_Control}) achieves the optimal solution, under any single communication link failure.
\end{corollary}

\begin{proof}
Let an arbitrary communication link $(i,j)$ fail. If there does not exist a power line between nodes $i$ and $j$, the communication topology is guaranteed to remain connected as it still contains a subtree of the power grid. Thus, the control mechanism will not be updated, and the optimal solution will be achieved. If there exists a power line between nodes $i$ and $j$, the control mechanism will be updated as in equation (\ref{FailedComm_Control}), which guarantees to achieve the optimal solution.
\end{proof}

Similar to the two-node system, one can find sufficient conditions under which the updated control mechanism in equation (\ref{FailedComm_Control}) is globally asymptotically stable for a multi-node system. For more details See Appendix \ref{MultiNode_StabilityProof}.

Consider Figure \ref{Toy_Example}, and suppose that the communication link between nodes $2$ and $7$ fail. Under the original control mechanism, the cost increases from $23.27$ to $35.69$. However, the new control mechanism will achieve the optimal solution. 

We compare the frequency response of the original control under full communication and the new control under single communication link failure. For simplicity, we only show the angular velocities at nodes $2$ and $7$ in Figures \ref{Original_Freq} and \ref{Updated_Freq_Single}; however, the same results hold for all the other nodes. We observed that for all nodes, the frequency response of the two control mechanisms are very similar, indicating that the new control mechanism will not create any abrupt changes in the frequency of the system.

\begin{figure}[ht]
\centering
\subfigure[Frequency Response of Original Control under Full Communication]
{\label{Original_Freq}\includegraphics[width=0.45\linewidth]{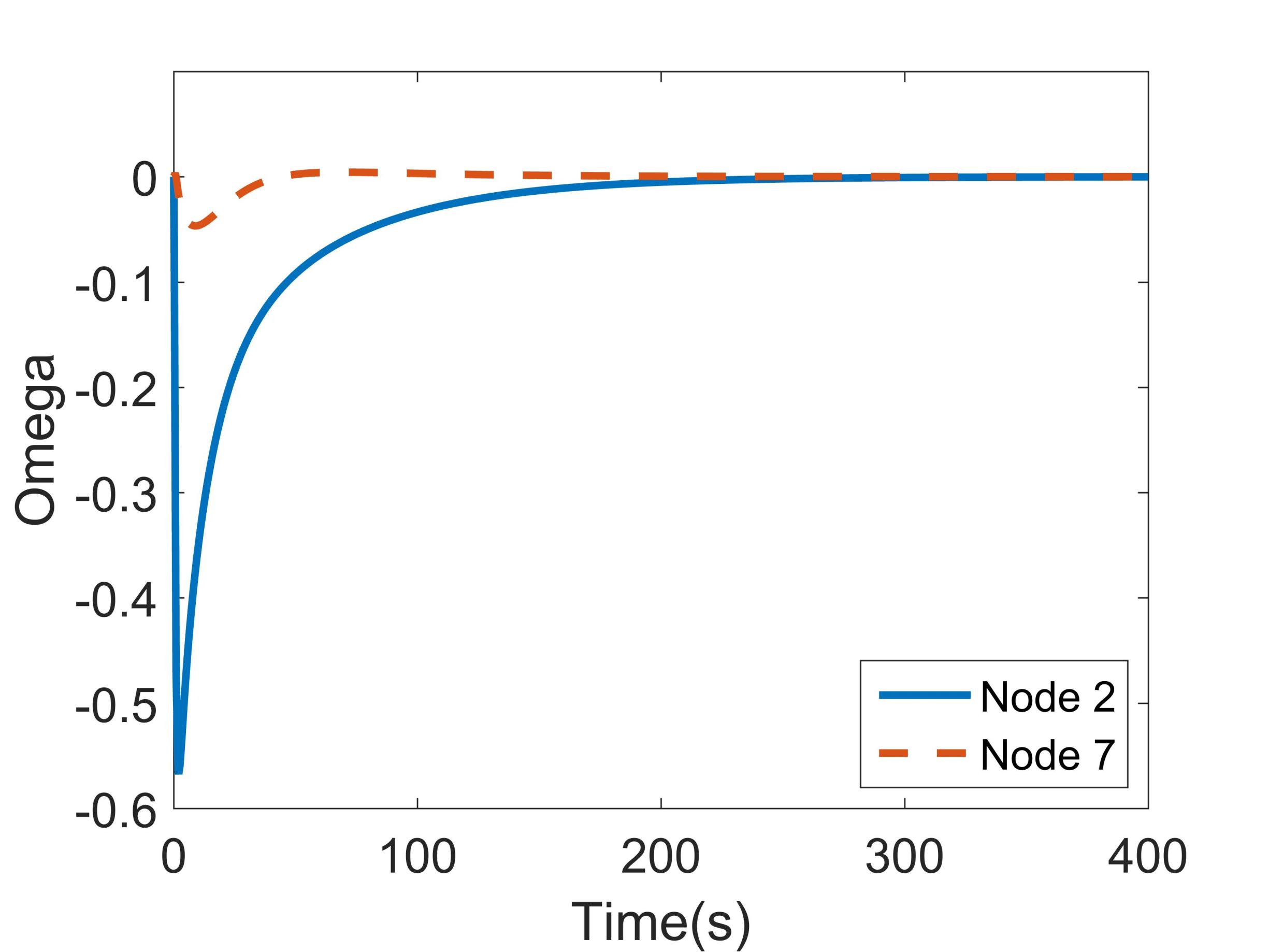}}                
\subfigure[Frequency Response of New Control under Single Link Communication Failure]
{\label{Updated_Freq_Single}\includegraphics[width=0.45\linewidth]{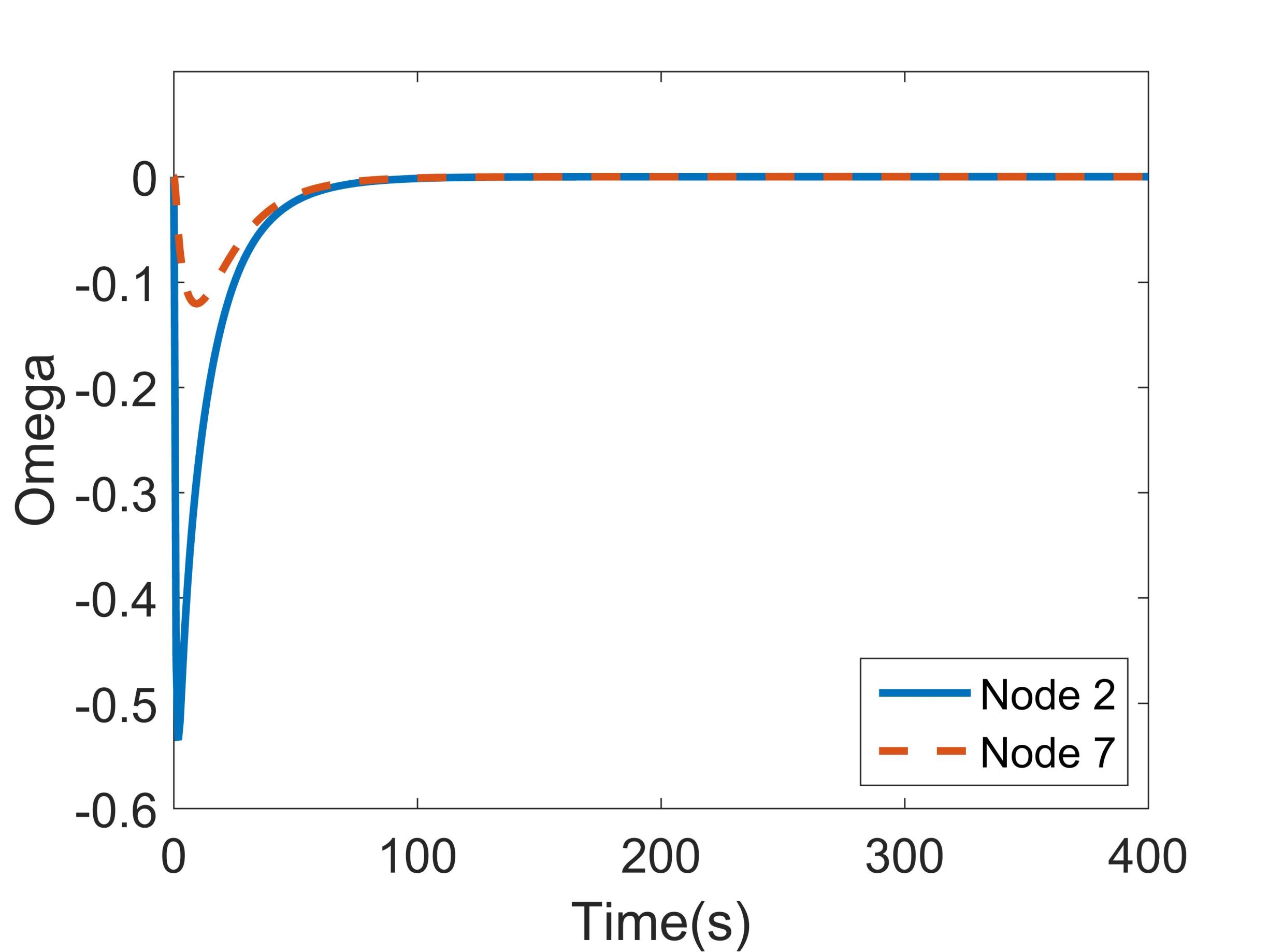}}
\caption{Comparing the frequency responses}
\end{figure}
\vspace{-5mm}

\subsection{Multiple Communication Link Failures}

In this section, we consider the case that multiple communication links fail (See Figure \ref{FailedComm_Multiple} as an example.) We generalize the control mechanism described for the single communication link failures as follows. 

\begin{figure}[ht]
\centering
\includegraphics[scale=0.5]{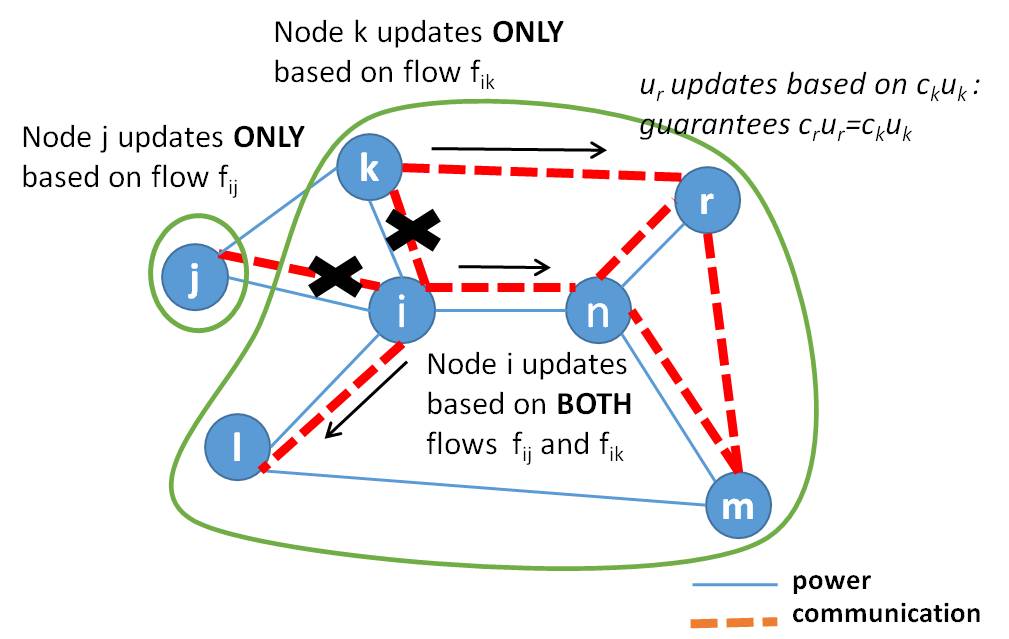}
\caption{Power Grid and Communication Network - Solid lines are power lines and dashed lines are communication lines.}
\label{FailedComm_Multiple}\vspace{-0.4cm}
\end{figure}

Consider pairs of nodes that have lost their communication links, but they are connected via power lines. Let $F$ be the set of such nodes. Moreover, let $q_i$ be the artificial variable for every node $i \in F$, and initialize it as $q_{i}(t_0)=\sum_{j \in F: (i,j)\in E_P} (C_i u_i(t_0) - C_j u_j(t_0))$. 

The update control rule can be written as follows.

\begin{subequations}
	\begin{alignat}{1}
		& C_r \dot{u_r}(t) = -\omega_r(t) - C_r \sum_{l: (r,l) \in \mathcal{E_C}}(C_r u_r(t) - C_l u_l(t)) \quad r \in \mathcal{N}\backslash F \\
 		& C_i \dot{u_i}(t) = -\omega_i(t) - q_{i}(t)   \quad i \in F\\
 		& \dot{q_{i}}(t) = -\sum_{j \in F: (i,j)\in E_P}(\omega_i(t) - \omega_j(t)) - 2q_{i}(t)   		
 	\end{alignat}
 	\label{MultipleFailedComm_Control}
\end{subequations}

It can be seen from equation (\ref{MultipleFailedComm_Control}) that every pair of node $i$ and $j$ that have lost their communication link, but are connected via a power line will switch to the new control rule, where the control rule at the rest of nodes remains the same. This control rule does not guarantee to achieve the optimal solution; however, we show that in practice it improves the cost.

Consider the power grid in Figure \ref{Toy_Example}, and assume that the communication links between nodes $1$ and $2$ and nodes $2$ and $5$ have failed. Under the original control, the cost increases from $23.27$ to $36.87$ which is $58\%$ increase in the optimal cost. However, our control described in equation (\ref{MultipleFailedComm_Control}) achieves a cost of $25.45$, which is only $9\%$ increase in the optimal cost ($49\%$ improvement). In addition, we observed that the new control policy will not lead to any unacceptable changes in the frequency response.

\section{Control with Discrete-Time Communication}\label{DiscreteComm_Sec}

In this Section, we study the impact of discrete-time communication on the performance of distributed frequency control. As discussed in Section \ref{Impact_Disc_Comm}, when the time interval between communication messages increases, the convergence time increases. In this Section, we propose an algorithm that sequentially updates the control of pairs of nodes using the dynamics of the power flow between them. Using simulation results, we show that the new algorithm converges much faster than the original one.

Let $T$ be the time interval between communication messages. Let $\mathcal{E_S}=\{ e_1,\cdots, e_m \} = \mathcal{E_P \cap E_C }$ be the set of pairs of nodes that share the power lines and communication links. The algorithm is as follows.

Let communication messages update at time instants $KT$, where $K \geq 0$. At each time interval $KT \leq t <(K+1)T$, the algorithm selects a link $e_r \in E_S$, and updates the control according to equations (\ref{FailedComm_Control}), where $i$ and $j$ are the end-nodes of the selected link $e_r$. The only difference is in equation (\ref{Original_control_singlefail}), where the control should be updated based on the most recent communication message received at time $KT$; i.e. $C_k \dot{u_k}(t) = -\omega_k(t) - C_k \sum_{l: (k,l) \in \mathcal{E_C}}(C_k u_k(t) - C_l u_l(KT)) \quad \forall k \neq i,j$. At the beginning of next time interval, new communication messages will be received, and the algorithm selects the next link in $E_S$. The algorithm keeps iterating on the links in sequence until convergence is achieved.

Figure \ref{Sequential_Update} shows the sequence of link selection and control updates at nodes. This algorithm improves the convergence rate because during each interval, it uses the additional information from the dynamics of the power grid to update the control at each node. 

\vspace{-2mm}
\begin{figure}[ht]
\centering
\includegraphics[scale=0.44]{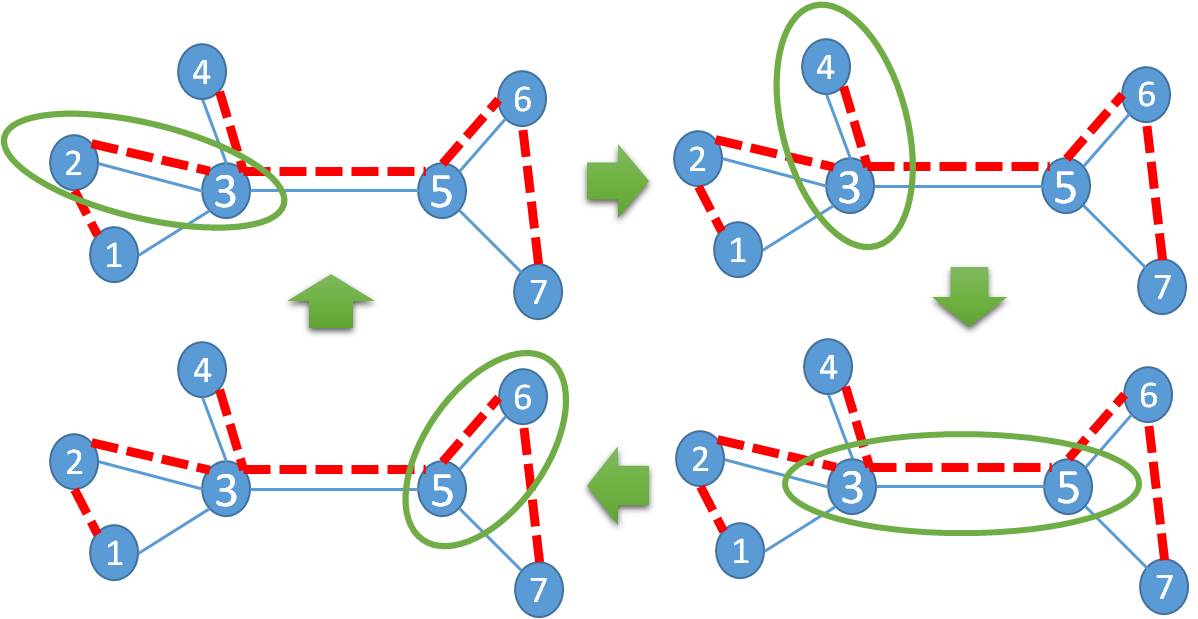}
\caption{Power Grid and Communication Network - Solid lines are power lines and dashed lines are communication lines. The shared edges between the power grid and communication network are $(2,3),(3,4),(3,5),(5,6)$, and the algorithm sequentially selects one of these edges, and uses its power flow to control the power changes at nodes.}
\label{Sequential_Update}\vspace{-3mm}
\end{figure}

We applied the original control scheme as well as the new control scheme to the power grid in Figure \ref{Toy_Example}. For simplicity, we only show the results for two nodes $1$ and $5$; however, the results are the same for the rest of nodes. Figures \ref{T=1ms} and \ref{T=1s} show that increasing the value of $T$ increases the convergence time under the original control. Figures \ref{T=1ms_OriginalControl} and \ref{T=1s_NewAlgorithm} indicate that by applying the new control mechanism, the convergence time for $T=1s$ is similar to the convergence time of the original control for $T=1ms$. In addition, it can be seen that although the general behavior of the power under both control mechanisms are similar, there are some fluctuations in the value of power under the new control algorithm. However, by comparing the frequency response of the control mechanisms in Figures \ref{Original_Freq_T=1ms} and \ref{Updated_Freq_T=1s}, it can be seen that the fluctuations in the frequency response of nodes under the new algorithm are negligible.

\vspace{-2mm}
\begin{figure}
\centering
\subfigure[T=1ms; Original Control]
{\label{T=1ms}\includegraphics[width=0.45\linewidth]{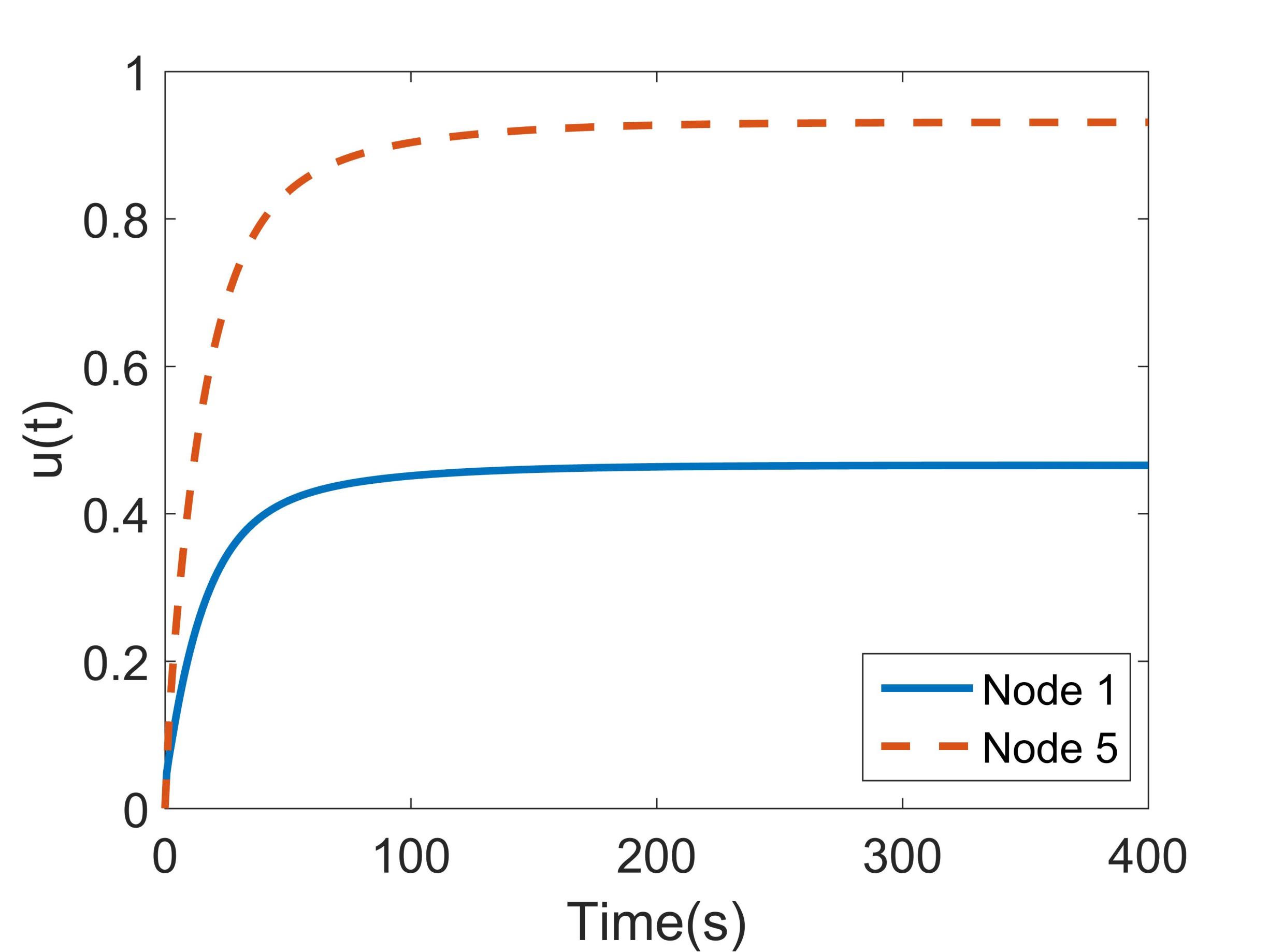}}                
\subfigure[T=1s; Original Control]
{\label{T=1s}\includegraphics[width=0.45\linewidth]{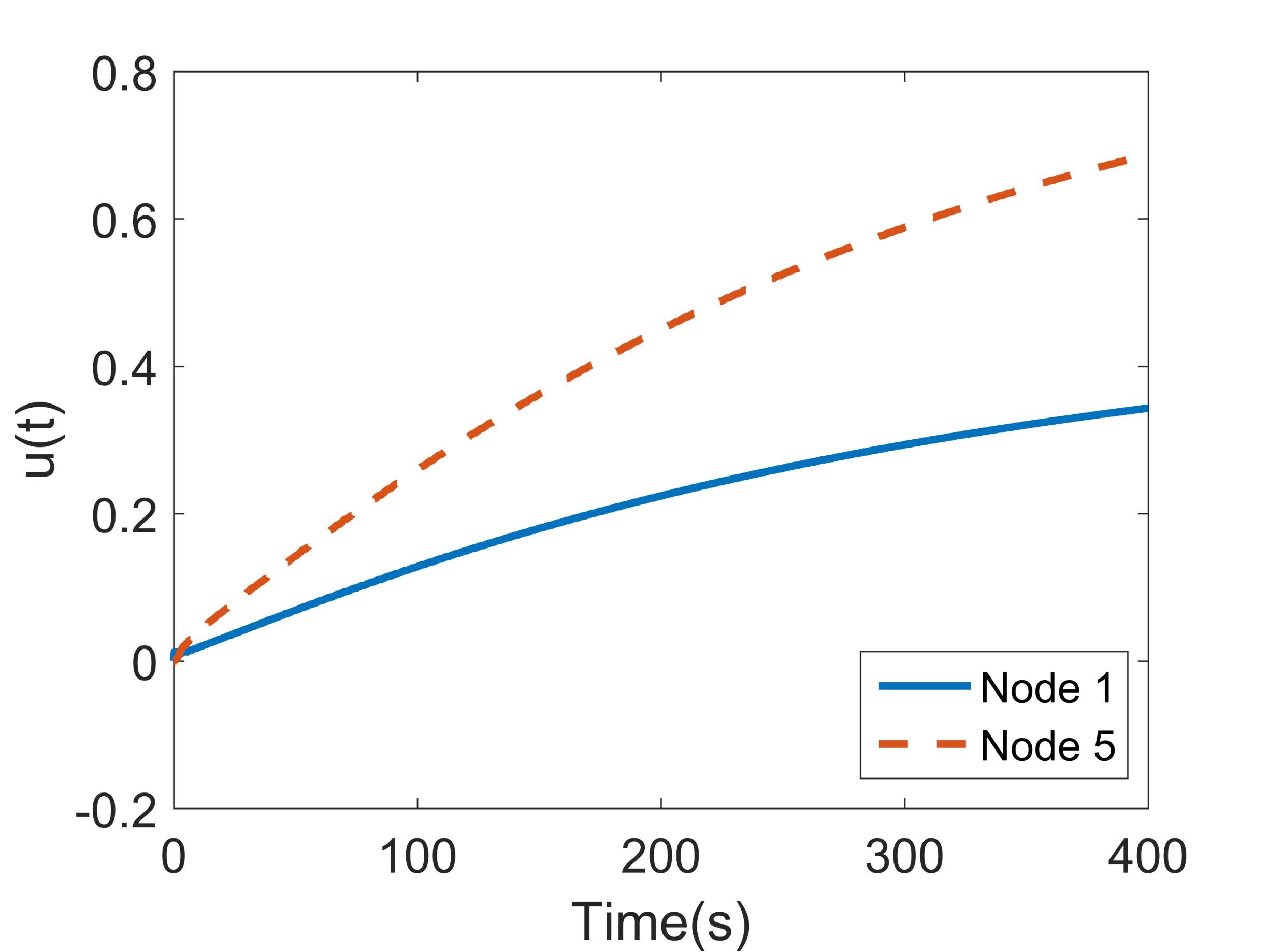}}
\subfigure[T=1ms; Original Control]
{\label{T=1ms_OriginalControl}\includegraphics[width=0.45\linewidth]{T=1ms}}                
\subfigure[T=1s; New Control Algorithm]
{\label{T=1s_NewAlgorithm}\includegraphics[width=0.45\linewidth]{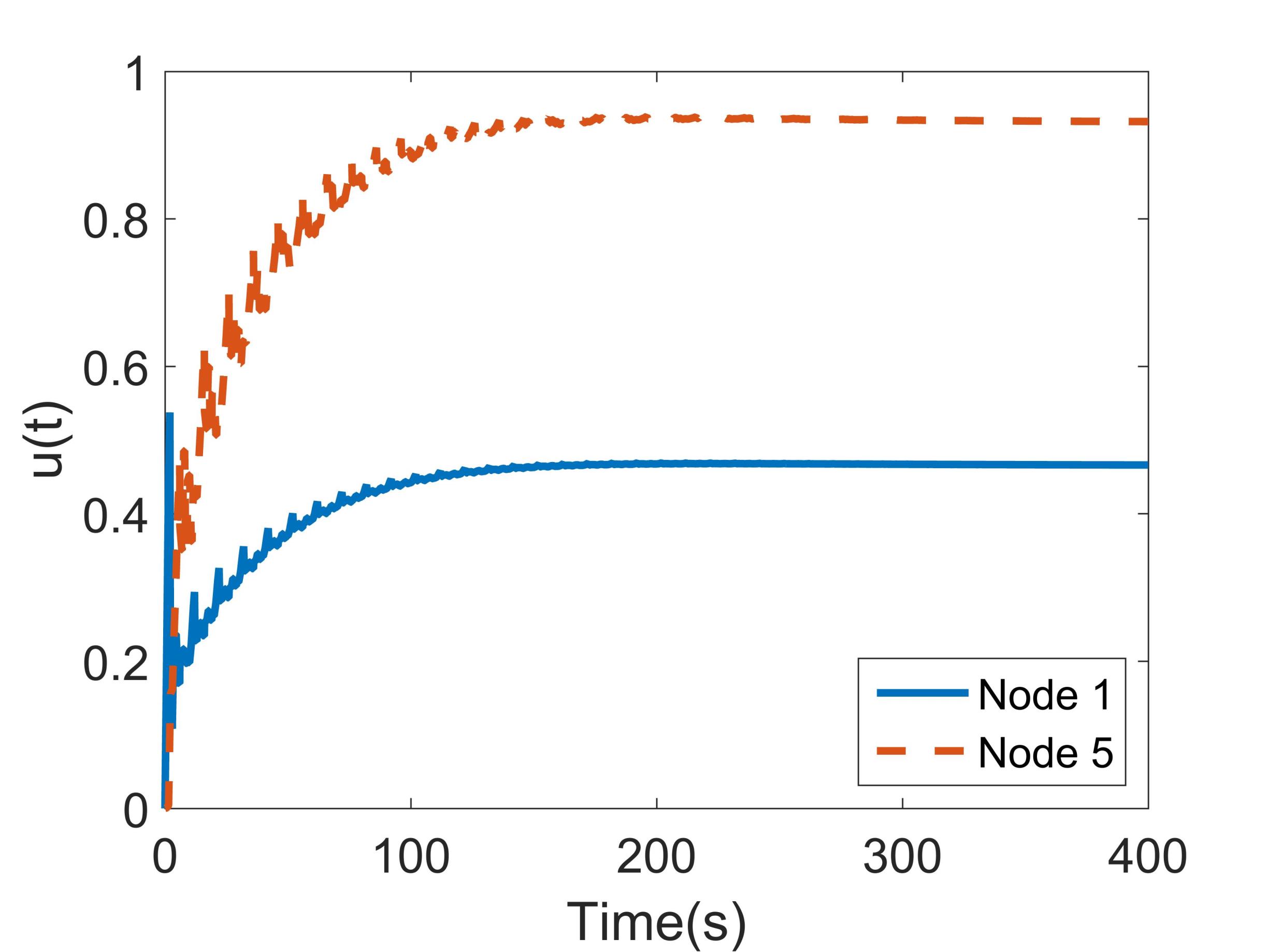}}
\subfigure[Frequency Response of Original Control for T=1ms]
{\label{Original_Freq_T=1ms}\includegraphics[width=0.45\linewidth]{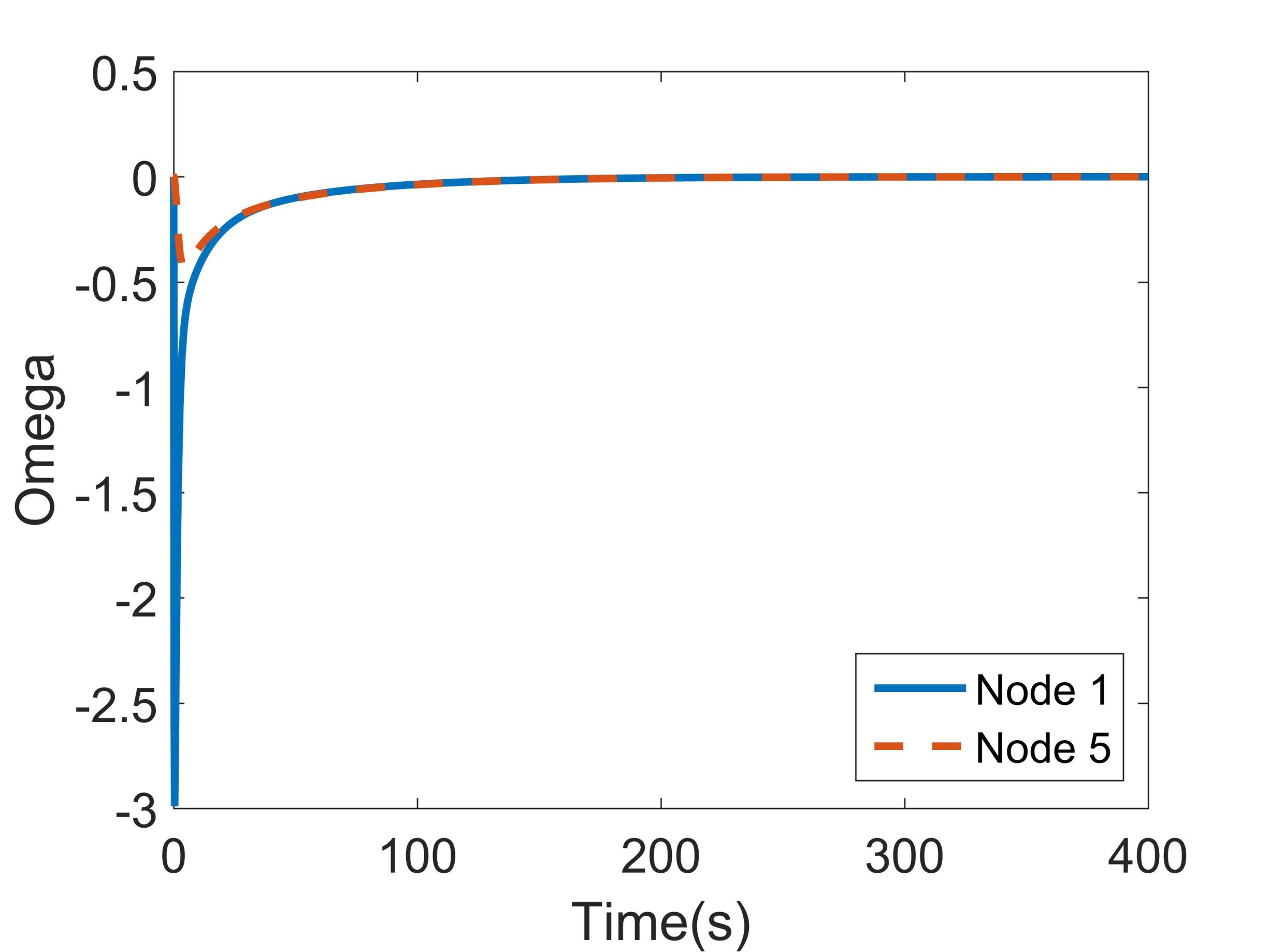}}                
\subfigure[Frequency Response of New Control for T=1s]
{\label{Updated_Freq_T=1s}\includegraphics[width=0.45\linewidth]{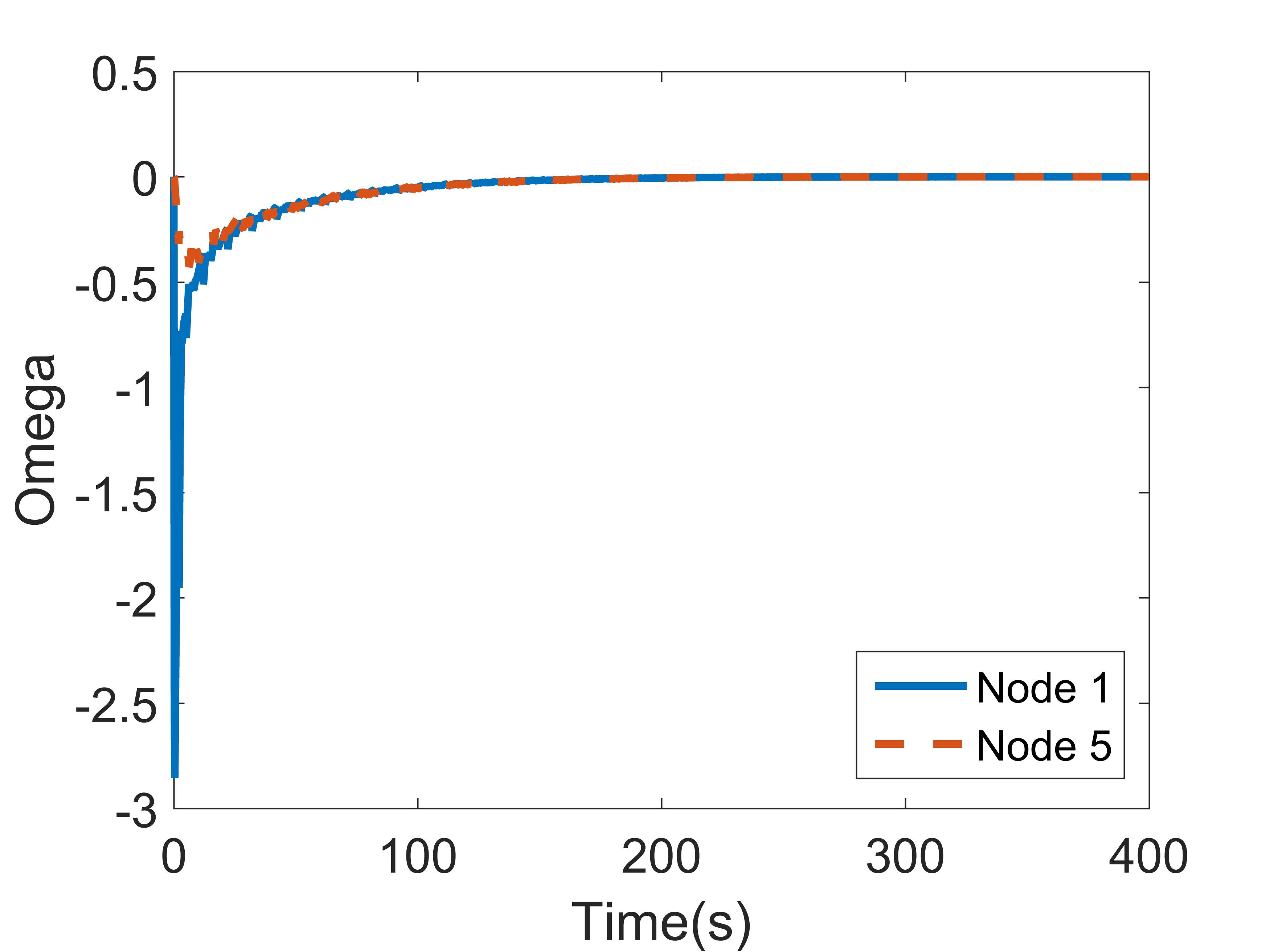}}
\caption{Comparing the power and frequency response for large T under new control with small T under the original control \vspace{-5mm}}
\end{figure}

\section{Conclusion}\label{Conclusion_Sec}
In this paper, we analyzed the impact of communication failures as well as discrete-time communication messages on the performance of optimal distributed frequency control. We considered the consensus-based algorithm proposed in \cite{zhao2015distributed} and \cite{dorfler2014breaking}, and showed that although the control mechanism can balance the power, it will not achieve the optimal solution under communication failures. 

Next, we proposed a novel control mechanism that uses the dynamics of the power flow between two nodes instead of the information received directly from the communication link between them. We proved that our algorithm achieves the optimal solution under any single communication link failure. We also used simulation results to show that the new control improves the cost under multiple communication link failures.

Finally, we showed that the convergence time of the distributed control increases as the time between two communication messages increases. We proposed a sequential control scheme which uses the dynamics of the power grid, and using simulation results, we showed that it improves the convergence time significantly.

\bibliographystyle{IEEEtran}
\bibliography{reference}

\appendices

\section{Data of the power grid in Figure \ref{Toy_Example}}\label{ToyPower_Data}

Inertia, initial power, droop control and cost of adjustable control at nodes 1 to 10 are as follows.

$M =[ 0.01 , 0.02 , 0.01 , 0.1 , 0.05 , 0.8 , 0.05 , 1 , 0.1 , 0.01]$

$P_0=[1 ,5 ,-2 ,6 ,-5 ,-10 ,-4 ,8 ,5 ,-4]$

$D = \frac{|P_0|}{3} $

$\quad \sim [0.33 ,1.67 ,0.67 ,2 ,1.67 ,3.33 ,1.33 ,2.67 ,1.67 ,1.33]$

$Cost = [10 ,10 ,100 ,100 ,5 ,10 ,7 ,9 ,5 ,10]$

Reactance of lines 1 to 10 are as follows.

$Reactance = [1 , 2 , 3 , 1 , 5 , 4 , 6 , 1 , 9 , 1]$

\section{Stability of Two-Node System}\label{Two-Node-Stability-Proof}
First, we show the process of simplifying the characteristic polynomial; and then, discuss the conditions under which the characteristic polynomial has negative roots.

\subsection{Simplifying $s(\lambda)$ for Two node system}\label{Two-Node-Simplify_Subsec}

Let the state vector of the two-node system be $[\omega_i, \omega_j, f_{ij}, u_i, u_j, q_{i}, q_{j}]$. We can rewrite the state matrix of our dynamical system as follows.

\vspace{5mm}
$A=
\begin{bmatrix}
-M^{-1}D  &-M^{-1}A_p 		& M^{-1}I^{2 \times 2}	& 0^{2 \times 2}\\
BA_p^T    & 0^{1 \times 1}	& 0^{1 \times 2}		& 0^{1 \times 2}\\
-C^{-1}   & 0^{2 \times 1} 	& 0^{2 \times 2}		& -C^{-1}\\
-L_c	  & 0^{2 \times 1}	& 0^{2 \times 2}		& -2I^{2 \times 2}
\end{bmatrix}
$\vspace{5mm}

Let $s(\lambda)=det(A-\lambda I)$ be the characteristic polynomial of matrix $A$. 

$
s(\lambda) \\
= \small{det
\begin{bmatrix}
-M^{-1}D-\lambda I^{2 \times 2} &-M^{-1}A_p 		& M^{-1}I^{2 \times 2}		& 0^{2 \times 2}\\
BA_p^T    						& -\lambda			& 0^{1 \times 2}			& 0^{1 \times 2}\\
-C^{-1}   						& 0^{2 \times 1} 	& -\lambda I^{2 \times 2}	& -C^{-1}\\
-L_c	  						& 0^{2 \times 1}	& 0^{2 \times 2}			& -(2+\lambda)I^{2 \times 2}
\end{bmatrix}}
$\vspace{5mm}

By schur complement formula, \\
$
s(\lambda)= (2+\lambda)^2 
det
\begin{bmatrix}
-M^{-1}D-\lambda I^{2 \times 2} 										&-M^{-1}A_p 			& M^{-1}\\
BA_p^T     																					& -\lambda				& 0^{1 \times 2}\\
C^{-1}[\frac{L_c}{2+\lambda}-I^{2 \times 2}]  			& 0^{2 \times 1}	& -\lambda I^{2 \times 2}

\end{bmatrix}
$\vspace{5mm}

Next, we take $(2+\lambda)^2$ into the matrix by multiplying the last row with $(2+\lambda)I^{2 \times 2}$. 

Thus,
$
s(\lambda) \\
= det
\begin{bmatrix}
-M^{-1}D-\lambda I^{2 \times 2} &-M^{-1}A_p 		& M^{-1}\\
BA_p^T     						& -\lambda			& 0^{1 \times 2}\\
C^{-1}[L_c-(\lambda+2)I]  		& 0^{2 \times 1}	& -\lambda(\lambda+2) I^{2 \times 2}

\end{bmatrix}
$\vspace{5mm}

Next, we take out $[L_c-(\lambda+2)I]$ from the big matrix as follows.\\

$
s(\lambda)= det(-[L_c-(\lambda+2)I]) \\
\small{det
\begin{bmatrix}
-M^{-1}D-\lambda I^{2 \times 2} &-M^{-1}A_p 		& M^{-1}\\
BA_p^T     						& -\lambda			& 0^{1 \times 2}\\
-C^{-1}  		& 0^{2 \times 1}	& \lambda(\lambda+2) [L_c-(\lambda+2)I]^{-1}

\end{bmatrix}}
$\vspace{5mm}

By simplifying the matrices, we will get the following.
$
s(\lambda)= \lambda(\lambda+2)det
\begin{bmatrix}
-M^{-1}D-\lambda I  &-M^{-1}A_p 		& M^{-1}\\
BA_p^T     			& -\lambda			& 0^{1 \times 2}\\
-C^{-1}  			& 0^{2 \times 1}	& -L_c-\lambda I
\end{bmatrix}
$\vspace{5mm}

By schur complement formula, \\
$
s(\lambda)= \lambda(\lambda+2) det(L_c + \lambda) \\
det
\begin{bmatrix}
-M^{-1}D-\lambda I - M^{-1}(L_c+\lambda I)^{-1}C^{-1}	&-M^{-1}A_p 	\\
BA_p^T     											& -\lambda		
\end{bmatrix}
$\vspace{5mm}

We apply the schur complement formula, one more time.\\
$
s(\lambda)= \lambda^2(\lambda+2) det(L_c + \lambda) \\
det
\begin{bmatrix}
-M^{-1}D-\lambda I - M^{-1}(L_c+\lambda I)^{-1}C^{-1}-\frac{1}{\lambda}M^{-1}A_pBA_p^T
\end{bmatrix}
$\vspace{5mm}

$
s(\lambda)= (\lambda+2) det(M^{-1})det(L_c + \lambda) \\
det
\begin{bmatrix}
\lambda D + \lambda^2 M + \lambda(L_c+\lambda I)^{-1}C^{-1}+L_p^B
\end{bmatrix}
$\vspace{5mm}

$
s(\lambda)= (\lambda+2) det(M^{-1}) \\
det
\begin{bmatrix}
\lambda (L_c + \lambda)D + \lambda^2 (L_c + \lambda)M + \lambda C^{-1}+ (L_c + \lambda)L_p^B
\end{bmatrix}
$\vspace{5mm}

$s(\lambda)= (\lambda+2) det(M^{-1}) det(H(\lambda))$, where $H(\lambda) = \begin{bmatrix}
(\lambda^2 D + \lambda^3 M + \lambda C^{-1}) + (\lambda L_cD+\lambda^2 L_c M + (2+\lambda)L_p^B)
\end{bmatrix}
$\vspace{5mm}

\subsection{Proof of Stability}\label{Two-Node-Stability_Subsec}

Roots of $s(\lambda)$ are $0$, $-2$ and roots $det(H(\lambda))$. Thus, it is enough to show that under the conditions in equations \ref{sufficient}, $det(H(\lambda))$ does not have a root on the right-hand side of the plane.

The necessary condition for $det(H(\lambda))=0$ is that there exists eigenvector $y \neq 0$ such that $H(\lambda)y = 0$. Therefore, $y^*H(\lambda)y = 0$. 

We show that under the conditions in equations \ref{sufficient}, for any $x \neq 0$, roots of $x^*H(\lambda)x = 0$ will be in the left-hand side of the plane; thus, it is a sufficient condition for the stability of our system.

Without loss of generality, we assume $x^*x=1$, and rewrite $x^*H(\lambda)x$ as follows.

$x^*H(\lambda)x =  a_0 + a_1 \lambda + a_2 \lambda^2 + a_3 \lambda^3 = 0$, where $a_0=x^*(2L_p^B)x$, $a_1=x^*(L_p^B + L_cD + C^{-1})x$, $a_2 = x^*(L_c M + D)x$ and $a_3 = x^*(M)x$.

Under the conditions in equations \ref{sufficient}, coefficients $a_1,a_2,a_3$ are all positive.

\begin{itemize}
	\item[-] $a_1=x^*(L_p^B + L_cD + C^{-1})x = x^*(L_p^B + \frac{1}{2}(L_cD + DL_c) + C^{-1})x$; $L_p^B$ is positive semidefinite, and $C^{-1}$ is positive definite. $\frac{1}{2}(L_cD + DL_c)$ is also positive definite by conditions in \ref{sufficient}.
	\item[-] $a_2 = x^*(L_c M + D)x = x^*(\frac{1}{2}(L_c M+M L_c) + D)x$; $D$ is a positive definite matrix; $\frac{1}{2}(L_c M+M L_c)$ is also positive definite by conditions in \ref{sufficient}.
	\item[-] $a_3 = x^*(M)x >0$, since $M$ is a positive definite matrix.
\end{itemize}

However, since $L_p^B$, the laplacian matrix of the power grid, is a positive semidefinite matrix, the coefficient $a_0$ will be nonnegative; i.e. $a_0=x^*(2L_p^B)x \geq 0$. We consider both cases where $a_0=0$ and $a_0>0$, and show that in each case, the roots of $det(H(\lambda))$ will be in the left-hand side of the plane.

\textbf{Case I:} Let $a_0=0$; Thus, $x^*H(\lambda)x =  a_1 \lambda + a_2 \lambda^2 + a_3 \lambda^3 = \lambda (a_1+a_2\lambda+a_3\lambda^2) = 0$. One root of the above equation is $\lambda = 0$, and since $a_1,a_2,a_3>0$, the other two roots will be in the left-hand side of the plane by the Routh-Hurwitz stability criteria.

\textbf{Case II:} Let $a_0>0$. By Routh-Hurwitz stability criteria, roots of $x^*H(\lambda)x$ will have negative real values, if $a_i>0$ for $i=0,1,2,3$, and $a_0 a_3 < a_1 a_2$.

$a_1a_2 = [x^*(L_p^B + L_cD + C^{-1})x ][x^*(L_c M + D)x] > [x^*(2 L_p^B)x][x^*(M)x] = a_0a_3$ if and only if 

$[\lambda_{min}(L_p^B + L_cD + C^{-1})][\lambda_{min}(L_c M + D)] > \lambda_{max}(2 L_p^B) \lambda_{max} (M)$

\section{Stability of Multi-Node System}\label{MultiNode_StabilityProof}
\subsection{Simplifying $s(\lambda)$ for Multi-Node System}

Let $D,M,C \in R^{N_P \times N_P}$ be the diagonal matrices of droop, inertia and cost values of the all nodes in the power grid.  Moreover, let $I \in R^{N_P \times N_P}$ be the identity matrix.

Suppose that we label the nodes such that nodes $N_P$ and $N_P-1$ be the nodes that have lost their communication link. Let $A_p$ be the adjacency matrix of the power grid. Moreover, let $L_c$ be the laplacian matrix of the communication network. Finally, define $ L_c^* =\begin{bmatrix} L_{c_1} \\0^{2\times N-2} | L_{c_2}C_2^{-1} \end{bmatrix}$, where \\
$L_{c_1}=L_c[1:N_P-2,1:N_P-2]$ and $L_{c_2}=\begin{bmatrix} 1 & -1 \\ -1 & 1 \end{bmatrix}$ be the laplacian matrix of a two-node system.

Then, the characteristic polynomial of our multi-node system is as follows.


$s(\lambda)=\lambda(\lambda+2)\\
det \begin{bmatrix}
-M^{-1}D-\lambda I^{N \times N}	  &-M^{-1}A_p 		& M^{-1}I^{N \times N}\\
BA_p^T    & -\lambda I^{E \times E}	& 0^{E \times N}	  \\
-C^{-1}   & 0^{N \times E} 	& -L^*_c C	-\lambda I^{N \times N}	
\end{bmatrix}
$\vspace{5mm}

Using the same techniques as in Section \ref{Two-Node-Simplify_Subsec}, the characteristic polynomial can be simplified as follows.

$
s(\lambda)= \lambda^{1+E}(\lambda+2) det(L^*_cC + \lambda) \\
\small{det
\begin{bmatrix}
-M^{-1}D-\lambda I - M^{-1}(L^*_cC+\lambda I)^{-1}C^{-1}-\frac{1}{\lambda}M^{-1}A_pBA_p^T
\end{bmatrix}}
$\vspace{5mm}



Let $L_p^B$ be the weighted laplacian matrix of the power grid, where $L_p^B = A_pBA_p^T$. 

Therefore, \\
$s(\lambda)=(-1)^{N} \lambda^{1+E-N}(\lambda+2) det(M^{-1}) det(H(\lambda))$, \\
where $H(\lambda)=
(\lambda^2 D + \lambda^3 M + \lambda C^{-1}) \\
+ (\lambda L^*_cCD+\lambda^2 L^*_c CM + (L^*_cC+\lambda)L_p^B)$
\vspace{5mm}

\subsection{Proof of Stability}

In this section, we claim that the following conditions are sufficient for the stability of the multi-node system.

\vspace{-5mm}
\begin{subequations}
	\begin{alignat}{1}
	&	M \succ 0  \label{cond1_Multi}\\
	&	\frac{1}{2}(L^*_cC M + MC L^{*T}_c) + D \succ 0 \label{cond2_Multi} \\
	&	L_p^B + \frac{1}{2}(L^*_cCD+DCL^{*T}_c) + C^{-1} \succ 0	 \label{cond3_Multi}\\
	&[\lambda_{min}(L_p^B + L^*_cCD + C^{-1})][\lambda_{min}(L^*_c CM + D)] \nonumber \\
	& > \lambda_{max}(L^*_c C L_p^B) \lambda_{max} (M) \label{cond4_Multi}
 	\end{alignat}
	\label{sufficient_Multinode}
\end{subequations}
\vspace{-5mm}

Similar to the two-node system, in order to prove the stability of the multi-node system, it is enough to show that the real-parts of all eigenvalues are negative. This is due to the fact that we have a linear system. Thus, we need to prove that the all the roots of $s(\lambda)$ are in the negative-side of the plane.

Since the power grid is a connected network, it contains a subtree; thus, $E\geq N-1$. Therefore, the roots of $s(\lambda)$ are $0$, $-2$ and roots $det(H(\lambda))$. Thus, it is enough to show that $det(H(\lambda))$ does not have a root on the right-hand side of the plane.

Similar to the two-node system, the necessary condition for $det(H(\lambda))=0$ is that there exists eigenvector $y \neq 0$ such that $H(\lambda)y = 0$. Therefore, $y^*H(\lambda)y = 0$.

We show that under conditions \ref{sufficient_Multinode}, for any $x \neq 0$, roots of $x^*H(\lambda)x = 0$ will be in the left-hand side of the plane; thus, these conditions are sufficient for the stability of our system.

Without loss of generality, we assume $x^*x=1$, and rewrite $x^*H(\lambda)x$ as follows.

$x^*H(\lambda)x =  a_0 + a_1 \lambda + a_2 \lambda^2 + a_3 \lambda^3 = 0$, where $a_0=x^*(L^*_c C L_p^B)x$, $a_1=x^*(L_p^B + L^*_cCD + C^{-1})x$, $a_2 = x^*(L^*_cC M + D)x$ and $a_3 = x^*(M)x$.

Similar to the two-node system, we first show that under conditions \ref{sufficient_Multinode}, coefficients $a_1,a_2,a_3$ are positive. 

\begin{itemize}
	\item[-] $a_1=x^*(L_p^B + L^*_cCD + C^{-1})x = x^*(L_p^B + \frac{1}{2}(L^*_cCD+DCL^{*T}_c) + C^{-1})x >0$. This is guarantees under condition \ref{cond3_Multi}.
	\item[-] $a_2 = x^*(L^*_cC M + D)x = x^*( \frac{1}{2} (L^*_cC M + MC L^{*T}_c) + D)x >0$. This is guarantees under condition \ref{cond2_Multi}.
	
	\item[-] $a_3 = x^*(M)x >0$, which is guaranteed under condition \ref{cond1_Multi}.
\end{itemize}

However, since $L_p^B$, the laplacian matrix of the power grid, is a positive semidefinite matrix, the coefficient $a_0$ will be nonnegative; i.e. $a_0=x^*(L^*_c C L_p^B)x \geq 0$. We consider both cases where $a_0=0$ and $a_0>0$, and find the sufficient conditions for each case under which the roots of $det(H(\lambda))$ are in the left-hand side of the plane.

\textbf{Case I:} Let $a_0=0$; Thus, $x^*H(\lambda)x =  a_1 \lambda + a_2 \lambda^2 + a_3 \lambda^3 = \lambda (a_1+a_2\lambda+a_3\lambda^2) = 0$. One root of the above equation is $\lambda = 0$, and since $a_1,a_2,a_3>0$, the other two roots will be in the left-hand side of the plane by the Routh-Hurwitz stability criteria.

\textbf{Case II:} Let $a_0>0$. By Routh-Hurwitz stability criteria, roots of $x^*H(\lambda)x$ will have negative real values, if $a_i>0$ for $i=0,1,2,3$, and $a_0 a_3 < a_1 a_2$.

$a_1a_2 = [x^*(L_p^B + L^*_cCD + C^{-1})x ][x^*(L^*_cC M + D)x] > [x^*(L^*_c C L_p^B)x][x^*(M)x] = a_0a_3$ if and only if 

$[\lambda_{min}(L_p^B + L^*_cCD + C^{-1})][\lambda_{min}(L^*_cC M + D)] > \lambda_{max}(L^*_c C L_p^B) \lambda_{max} (M)$


\end{document}